\documentclass[sigconf,nonacm]{acmart}
\setcopyright{none}
\acmConference[]{}{}{}
\acmYear{2026}
\copyrightyear{2026}
\acmDOI{}
\acmISBN{}
\usepackage{amsmath,booktabs,graphicx,url,array}

\hypersetup{hidelinks,pdftitle={Hierarchical Rule Auditing with Exact Ternary GPU Correlation},pdfauthor={Shengwei Wang}}
\usepackage{algorithm,algpseudocode}
\newtheorem{auditproposition}{Proposition}
\newcommand{\botv}{\bot}
\newcommand{\Ind}{\mathbf{1}}

\title{Hierarchical Rule Auditing with Exact Ternary GPU Correlation}
\author{Shengwei Wang}

\begin{document}
\begin{abstract}
Auditing firewall logs against requirements mapped to controls and frameworks must retain overlapping obligations, preserve uncertainty from missing attributes and count distinct affected records. Direct, bitmap and relational methods can express this audit; a correlation implementation must preserve distinctions that ordinary match counts and rule-hit totals discard. We present an exact ternary correlation algorithm that compiles categorical conjunctions into one bounded-integer matrix product. A dominating false-predicate penalty separates known mismatches from missing evidence, while per-record hierarchy unions prevent duplicate counting. We prove their composition and handle high-cardinality endpoints through exact address classes and versioned labels. Evaluation compares eight CPU/GPU implementations, includes public traffic and 133 published source-address ACLs with generated observations, and measures hierarchy and policy-rebuild costs. On an NVIDIA L4, matched fused execution reduces latency by 8.8 to 11.2 percent relative to two-product correlation across five tested shapes. At 100,000 records and 1,024 requirements, it takes 12.28~ms versus 13.82~ms for two products and 12.29~ms for GPU bitmaps. Bitmaps remain faster in most regimes, with hierarchy operations and reductions dominating the measured fused-correlation stage intervals.
\end{abstract}
\maketitle

\section{Introduction}
Organizations audit recorded firewall decisions against security requirements to identify connections that conflict with a supplied policy. The result must connect findings to controls and compliance frameworks, distinguish confirmed mismatches from unresolved applicability, and report both obligation hits and distinct affected records. Given a batch of logs and a versioned policy snapshot, the task is to evaluate every applicable independent obligation and aggregate its findings across overlapping reporting groups. Losing uncertainty or double-counting a record changes the answer, even when the underlying packet decision was recorded correctly.

For example, one allowed connection to a restricted database may violate both a boundary obligation and a database-specific obligation. A framework then has two requirement hits but one affected record. If the destination zone is missing, applicability may instead be unresolved. A control shared by two frameworks adds another overlap: each framework can report the record, but their distinct counts cannot simply be added to count records across frameworks.

Several established approaches can implement this audit. Packet classifiers and hardware match/action pipelines select actions and collect traffic metadata~\cite{packetclassification,gpuclassifier,docaflow}. Hash indexes and compiled bitmaps reduce candidate search or predicate work~\cite{dpdkhash,packetclassification}. Policy languages and relational engines can express conditions, joins and aggregates~\cite{cedar,opa,duckdb,gpujoin}. Logical matrix representations and tensor compilers establish weighted tests, thresholded scores and analytical operators as building blocks~\cite{logicmatrix,hummingbird,tqp}. These techniques motivate both the correlation approach and the direct and bitmap alternatives used in our evaluation.

However, translating this audit into correlation and aggregate operations exposes two information losses. An observed-match count can assign the same score to a missing constrained value and a known disallowed value, although their applicability states differ. Separately, per-rule totals discard the record overlap needed for distinct counts at controls and frameworks. Our two-product formulation resolves the first ambiguity with observed and possible match scores, but introduces separate score products and a correction. Preserving the required information in one product is the algorithmic question; whether that change reduces complete-audit latency is a separate performance question, because encoding and hierarchy reductions also consume time.

Our insight is to encode a conjunction through penalties rather than positive match counts. A true field contributes zero, an unknown field one, and a false field a weight greater than the number of fields. A single record--requirement product then separates all three states: no false field can be hidden by the unknown contributions. Extended one-hot field blocks implement this sum for categorical allowed sets. The bank of record--requirement correlations draws inspiration from CDMA receivers~\cite{tse}; its logical meaning is established by the encoding and thresholds.

Turning this insight into an exact, measurable audit requires three technical steps:
\begin{itemize}
\item \textbf{Preserve predicate semantics in one score.} Wildcards, empty allowed sets and missing values require different treatment. We construct the ternary weights, prove equivalence to the observed/possible formulation and bound intermediate integers so that the specified GPU arithmetic is exact (Section~\ref{sec:single}).
\item \textbf{Compose matching with compact inputs and overlapping reports.} IP addresses must not require one coordinate per observed host, and shared hierarchy paths must not duplicate findings. We integrate exact address classes and versioned labels, and prove that Boolean reachability and per-record unions produce the complete report under disjoint tiling. We also separate policy rebuilding, record re-encoding and historical re-audit costs (Section~\ref{sec:lifecycle}).
\item \textbf{Measure the benefit at the complete-report boundary.} We implement matched fused kernels for both matrix formulations and compare eight CPU/GPU implementations returning the same ten vectors. Controlled regimes, public traffic, published ACLs, hierarchy variation and policy-rebuild experiments expose costs beyond the score products. Stage diagnostics separate the effects of score evaluation and reporting work.
\end{itemize}

On the L4, the fused single-product method reduces complete prepared-audit latency by 8.8 to 11.2 percent relative to fused two-product execution across five tested shapes. At 100,000 records and 1,024 requirements, the two methods take 12.28 and 13.82~ms; GPU bitmaps take 12.29~ms, within observed variability of the single-product method. Bitmaps remain faster in most regimes. The contribution is an exact single-product construction with proven report composition and evidence locating its practical benefit: reducing score boundaries helps the matrix path, while hierarchy processing and launch overhead limit the improvement to the full audit.

\section{Problem Context and Existing Approaches}\label{sec:context}
\paragraph{Packet matching and hardware pipelines}
Multi-field classification has established range-search, bitmap, tuple-space and GPU approaches~\cite{packetclassification,gpuclassifier}. NuevoMatch combines learned range indexing with rule validation~\cite{nuevomatch}; TaNG combines a learned model with tuple-space search and GPU execution~\cite{tang}. A typical classification result selects a rule or action; the present audit instead retains every applicable independent obligation and unresolved case before reporting.

NVIDIA DOCA Flow exposes hardware-accelerated match/action pipes on BlueField, with packet modifications, forwarding or dropping, counters and metadata~\cite{docaflow}. A DPU can therefore contribute telemetry and is not restricted to firewall enforcement. The comparison boundary is the requested computation: a packet-action or counter result does not directly supply retrospective ternary applicability and overlapping hierarchy unions over a chosen log snapshot. A DPU implementation of that complete contract would be a relevant alternative. We do not benchmark one or infer a GPU advantage from packet-forwarding rates.

\paragraph{Hash indexes and direct predicates}
An exact-key hash lookup can reduce the candidates examined. DPDK's hash library uses a configurable hash for bucket selection and signatures, then compares full keys~\cite{dpdkhash}; equal signatures alone do not establish key equality. Prefixes, wildcards and allowed sets additionally require predicate-aware indexing or evaluation. Our direct and bitmap baselines evaluate predicates without cryptographic hashes. Comparing a hash-indexed classifier would additionally require its candidate-search stage under the same audit semantics. In this implementation, SHA-256 fingerprints identify policy/input versions for provenance, separately from predicate evaluation.

\paragraph{Policy engines and log analytics}
Cedar formalizes authorization semantics and differential validation~\cite{cedar}; OPA evaluates general structured policies~\cite{opa}. Such languages can express richer conditions than this fixed kernel. A log-analysis program can evaluate predicates, join rule mappings and aggregate records using relational operators. GPU join research and DuckDB provide relevant execution foundations~\cite{gpujoin,duckdb}. These are viable implementation routes, not evidence that the audit cannot be expressed elsewhere. The distinguishing contract here is explicit missingness, independent obligations and two kinds of hierarchy count. VeriFlow checks network-wide invariants from forwarding rules and topology~\cite{veriflow}, a different input and result from this log audit.

\paragraph{Matrix compilation and the remaining design question}
Tensor encodings of predicates and logical connectives predate this work~\cite{tensorlogic}. Sakama et al. encode Boolean conjunctions using matrix weights and thresholding in a logic-program fixpoint computation~\cite{logicmatrix}. Hummingbird's GEMM strategy compiles numerical tests into path scores, leaf selection and predicted classes~\cite{hummingbird}. TQP++ supplies compiler-based analytical execution, including predicated null handling and data-dependent joins and aggregation~\cite{tqp}. These precedents establish logical matrix representations and missing-value support. Building on them, our construction lowers ternary categorical conjunctions to one bounded-integer score product and composes it with outcome gates and hierarchy unions. It improves on our observed/possible two-product formulation; comparison with nonmatrix alternatives remains empirical, without an asymptotic improvement claim. Section~\ref{sec:cdma} states the limited CDMA connection.

\paragraph{Audit scope and supplied evidence}
Every applicable organizational obligation is evaluated independently against a recorded binary outcome. Overlapping allow and deny obligations can both apply; first-match and deny-overrides enforcement are different contracts. Section~\ref{sec:stanford} uses a bounded priority-preserving translation solely for the standard-ACL evaluation. Frameworks supply reporting mappings, not additional packet predicates. NIST AC-4, SC-7 and CM-7 motivate the authored examples~\cite{nist}; those mappings encode a supplied interpretation, not a certification or automatic translation of a framework's prose.

Policy, exemption, hierarchy and inventory facts are supplied inputs. Correctness is conditional on them and the recorded outcomes. The interface validates values and identifiers; fingerprints associate the report with its inputs but do not authenticate evidence or establish collection completeness. Policy approval, log collection and evidence custody belong to the surrounding deployment.

\section{Snapshot Semantics}\label{sec:model}
Let $X$ contain $N$ records with $K$ categorical fields, each in a finite domain $D_k$ or the missing value $\botv$. There are $R$ requirements, $C$ controls and $Q$ frameworks. The implementation uses $K=8$, with the first field identifying the tenant; $t(i)$ denotes that category or its missing sentinel. Domains are fixed by the supplied policy snapshot. The other fields select declared observation values, inventory labels or exact address classes, as described in Section~\ref{sec:endpoints}. Appendix~\ref{app:schema} catalogs the supported schema and example layouts. Record identifiers are unique within a snapshot. Each observed outcome $o_i\in\{0,1\}$ is not-permitted or permit, respectively. An input adapter must explicitly resolve other outcomes or reject them. Brackets denote indicator functions, applied elementwise to matrices; zero and one represent false and true.

Requirement $j$ gives each field either a wildcard or an allowed set $A_{jk}\subseteq D_k$. It also supplies an intended outcome $a_j$. Let $c_{jk}$ indicate a constrained field. Define field evaluation $p_{ijk}$ by
\[
p_{ijk}=\begin{cases}
\mathsf{T},&c_{jk}=0,\\
\mathsf{F},&c_{jk}=1,\ A_{jk}=\varnothing,\\
\mathsf{U},&c_{jk}=1,\ X_{ik}=\botv,\ A_{jk}\ne\varnothing,\\
\Ind[X_{ik}\in A_{jk}],&\text{otherwise}.
\end{cases}
\]
A conjunction is false if any component is false, true if all are true, and unknown otherwise. Thus an all-wildcard requirement matches even an entirely missing record. A constrained empty set never matches, including on missing evidence. A disjunction containing several acceptable values is one field predicate, not several required simultaneous matches.

Let $E_{jt}$ be an explicit exemption relation over requirements and tenant categories, including a separate missing-tenant entry. A missing tenant does not inherit another category's exemption. Let $g_{ij}=[o_i\ne a_j]\land\neg E_{j,t(i)}$. The binary matrices $V,U\in\{0,1\}^{N\times R}$ record confirmed violations and indeterminate applicability relevant to an outcome mismatch:
\begin{align}
V_{ij}&=[ (\bigwedge_k p_{ijk})=\mathsf{T}]\land g_{ij},\\
U_{ij}&=[ (\bigwedge_k p_{ijk})=\mathsf{U}]\land g_{ij}.
\end{align}
$V$ and $U$ are disjoint at the record--requirement level. A record can nevertheless have a confirmed finding for one requirement and an indeterminate finding for another. A matching intended outcome is not reported as a violation merely because an applicability field is absent. Overall missing-record counts are exposed separately.

\subsection{What ``distinct'' means}
Let $H\in\{0,1\}^{R\times C}$ map requirements to controls, and $G\in\{0,1\}^{C\times Q}$ map controls to frameworks. Mappings may be many-to-many. Compile rule-to-framework reachability as $B=[HG>0]$, applying a Boolean threshold before report aggregation. Without this step, multiple paths become accidental weights. A dense compilation costs $O(RCQ)$ once per policy snapshot; it computes reachability without enumerating individual paths.

For either class $Z\in\{V,U\}$, the engine returns five vectors:
\begin{align}
q_Z&=\mathbf{1}^{\mathsf T}Z,&
h_Z^{c}&=q_ZH,&h_Z^{f}&=q_ZB,\\
d_Z^{c}&=\mathbf{1}^{\mathsf T}[ZH>0],&&
d_Z^{f}=\mathbf{1}^{\mathsf T}[ZB>0].
\end{align}
Here $q_Z$ counts requirement hits, $h_Z$ counts hits associated with a hierarchy node, and $d_Z$ counts distinct affected record identifiers. Ten output vectors preserve both classes. The class-specific distinct sets may overlap at a control or framework, so their counts must not be added as disjoint populations.

\paragraph{Hierarchy size and cross-framework overlap}
This is a three-level directed graph with at most $RC+CQ$ declared edges, not an enumeration of rule combinations. For example, two paths from one rule through different controls to the same framework produce one entry in $B$. If a record violates two requirements that reach that framework, its hit count is two and its distinct count is one. A shared record may also appear once in another framework; their distinct counts cannot be summed to recover a union. The engine reports the declared $Q$ framework groups rather than all $2^Q$ subsets. It neither infers semantic equivalence between controls nor evaluates arbitrary framework intersections or recursively nested policies. Those require an explicit mapping or a separate query contract.

Distinct counts cannot generally be reconstructed from $q_Z$. Two requirements with one hit each may affect the same record or different records. Their marginal totals are identical while the distinct count differs. An implementation must therefore retain incidence information long enough to compute a union, or use an equivalent per-record reduction.

\begin{auditproposition}[Disjoint snapshot aggregation]\label{prop:batches}
For a fixed requirement and hierarchy snapshot, all ten vectors are additive across batches with disjoint record identifiers.
\end{auditproposition}
\begin{proof}
Each output is a sum over records of a per-record hit count or membership indicator. Partitioning those records into disjoint batches partitions the sum. Overlapping identifiers or different policies invalidate this argument.
\end{proof}
The CLI checks uniqueness within one request. It does not implement deduplication across requests or distributed exactly-once ingestion.

\section{Exact Encodings and Arithmetic}\label{sec:encoding}
Partition a feature space of width $D=\sum_k |D_k|$ into disjoint field blocks. This concatenation does not enumerate the Cartesian product of attribute domains. Present values have one one-hot bit in their block of $F\in\{0,1\}^{N\times D}$; missing values have a zero block. The requirement matrix $\mathbf{R}\in\{0,1\}^{R\times D}$ sets bits for allowed categories in constrained fields and uses zero wildcard blocks. Let $\kappa_j=\sum_k c_{jk}$ count constrained fields, and let $L_{jk}$ indicate a constrained nonempty allowed set. Define the observed score $S$ and possible score $P$:
\[
S=F\mathbf{R}^{\mathsf T},\qquad W_{ik}=[X_{ik}=\botv],\qquad P=S+WL^{\mathsf T}.
\]

\begin{auditproposition}[Categorical score contract]\label{prop:score}
Each pair satisfies $0\le S_{ij}\le P_{ij}\le\kappa_j$. Applicability is true iff $S_{ij}=\kappa_j$, unknown iff $S_{ij}<P_{ij}=\kappa_j$, and false iff $P_{ij}<\kappa_j$.
\end{auditproposition}
\begin{proof}
A constrained field adds $(1,1)$ to $(S_{ij},P_{ij})$ if true, $(0,1)$ if unknown, and $(0,0)$ if false. The second case is supplied by $WL^{\mathsf T}$ for a missing value with a nonempty allowed set. An empty allowed set is always false. Wildcards contribute to neither score and are absent from $\kappa_j$. Summing these contributions gives the bounds. Reaching $\kappa_j$ in $S$ requires every constraint to be true. Reaching it only in $P$ requires no false constraint and at least one unknown; failing to reach it in $P$ requires a false constraint. The all-wildcard case has $S=P=\kappa_j=0$ and is true.
\end{proof}
The threshold $\kappa_j$ counts constrained fields, not set bits in $\mathbf{R}_j$: multiple allowed values within a field represent alternatives. A wildcard's truth is independent of evidence availability. The identity assumes single-valued fields and collision-free block dictionaries. After applicability is determined, the observed-outcome and exemption gates determine whether the record contributes to $V$ or $U$.

\paragraph{One control and one observed flow}
NIST SP~800-53 CM-7(b), Least Functionality, calls for prohibiting or restricting organization-defined ports and other system capabilities~\cite{nist}. For this illustration, instantiate its port restriction as ``deny destination port 23,'' with every other field a wildcard and no exemption. Data row 1611 of the public UCI firewall trace records source port 28441, destination port 23, action \texttt{allow}, two packets and 152 bytes~\cite{uci}. In the replay's destination-port dictionary $(23,3389,445,22,443,\mathrm{other})$, the record and requirement blocks are
\[
f_{\mathrm{port}}=r_{\mathrm{port}}=(1,0,0,0,0,0).
\]
All other requirement blocks are zero, giving
\[
S_{ij}=\langle F_i,\mathbf{R}_j\rangle=1=\kappa_j.
\]
Thus Proposition~\ref{prop:score} declares true applicability. The seven unavailable categorical fields are wildcards, so their missingness adds no unresolved constraint: $(WL^{\mathsf T})_{ij}=0$. The observed \texttt{allow} differs from the required deny, yielding $V_{ij}=1$ and $U_{ij}=0$. The port selection is the example's parameter choice; the source institution's policy is unavailable.

\subsection{One product for all three applicability states}\label{sec:single}
\paragraph{Why the observed score is insufficient}
Consider a requirement that constrains one field to $\{a\}$ and wildcards every other field. A record with that field missing and a record with a known value $v\ne a$ both have $S=0$ and $\kappa=1$. The first has unknown applicability and the second false applicability. Thus no decision function of $(S,\kappa)$ alone can recover the three-valued result for every valid input. Proposition~\ref{prop:score} retains the distinction through the possible score; the following construction retains it in an additional coordinate and a dominating false weight.

Augment each field block with one coordinate for its missing value. Every row of $F^+\in\{0,1\}^{N\times(D+K)}$ has exactly one active coordinate per field, including absent fields. Set $b=K+1$. For requirement $j$, define a weight for each possible value $v$, including $\botv$:
\[
\Theta_{jk}(v)=\begin{cases}
0,&c_{jk}=0,\\
b,&c_{jk}=1,\ A_{jk}=\varnothing,\\
1,&c_{jk}=1,\ v=\botv,\ A_{jk}\ne\varnothing,\\
0,&c_{jk}=1,\ v\in A_{jk},\\
b,&\text{otherwise}.
\end{cases}
\]
Concatenate the weights into $\Theta\in\{0,1,b\}^{R\times(D+K)}$ and compute
\[
\Xi=F^+\Theta^{\mathsf T}.
\]

\begin{auditproposition}[Single-product ternary correlation]\label{prop:single}
Let $f_{ij}$ and $u_{ij}$ count false and unknown constrained fields. Then
\begin{align}
\Xi_{ij}&=b f_{ij}+u_{ij}\\
&=b(\kappa_j-P_{ij})+(P_{ij}-S_{ij}).
\end{align}
Applicability is true iff $\Xi_{ij}=0$, unknown iff $0<\Xi_{ij}<b$, and false iff $\Xi_{ij}\ge b$.
\end{auditproposition}
\begin{proof}
Exactly one coordinate is selected in each field block. Its weight is zero for a true field, one for an unknown field and $b$ for a false field. Wildcards contribute zero and empty constrained sets contribute $b$, including on missing evidence. Summing yields the first identity. Proposition~\ref{prop:score} gives $f_{ij}=\kappa_j-P_{ij}$ and $u_{ij}=P_{ij}-S_{ij}$, yielding the second. With no false field, the score is $u_{ij}\le K<b$: it is zero precisely when every field is true. Any false field contributes at least $b$, and all other contributions are nonnegative.
\end{proof}
For unit unknown weights and a decision threshold equal to the integer false weight $b$, the smallest weight that handles all $K$ fields being unknown is $K+1$. This is a bound within this encoding family, not a lower bound on all representations of three-valued logic.

In the observed port-23 example, the allowed value has weight zero, as do the seven wildcard fields, so $\Xi_{ij}=0$. If that constrained port were absent, $\Xi_{ij}=1$ would indicate unknown applicability. If a different known port were observed, $\Xi_{ij}=9$ would indicate false applicability. The outcome and exemption gates still determine whether a true or unknown match is reported.

The single product replaces the two products of Proposition~\ref{prop:score}. Their combined inner dimension is already $D+K$, so this transformation does not reduce the asymptotic arithmetic count. It removes a score-product boundary and a separate score correction while extending the one-hot encoding to missing coordinates. The experiments measure whether that change improves complete execution.

\subsection{From CDMA correlation to rule scores}\label{sec:cdma}
CDMA correlation receivers evaluate signal--signature inner products to recover transmitted symbols~\cite{tse}. This motivates evaluating a bank of record--requirement correlations: $S_{ij}=\langle F_i,\mathbf{R}_j\rangle$, or $\Xi_{ij}=\langle F_i^+,\Theta_j\rangle$ in the ternary construction. Propositions~\ref{prop:score} and~\ref{prop:single} specify how these scores recover logical applicability.

Each audit input directly encodes one record rather than a noisy superposition of transmitted symbols. Overlapping requirement templates can all match; their mutual cross-correlations do not enter a record--requirement score. The construction therefore requires neither orthogonal signatures nor interference cancellation. The connection is the bank of correlations, not an isomorphism of the complete communication and audit systems. Hierarchical reporting subsequently follows supplied edges in $H$ and $B$; those edges are not spreading-code relationships.

\subsection{Finite-precision bounds}
The positive scores lie in $[0,K]$ and the single-product score in $[0,K(K+1)]$. With $K=8$, its weights and every partial sum are integers between zero and 72. Binary16 represents these exactly. Per-record hierarchy products sum at most $R\le1024$ binary contributions, also below binary16's consecutive exact-integer bound of 2048. Final reductions use signed 64-bit integers, and valid inputs must keep their aggregate counts in range. Reducing millions of records directly in binary16 would violate these bounds.

These arguments establish arithmetic exactness for the specified input limits and products. CuPy dispatches the products to its installed library. Instruction profiles separately establish which products actually use Tensor Cores on the measured device; a floating-point API alone does not demonstrate hardware execution.

\section{Implementation}
\subsection{High-cardinality endpoints: labels and subnets}\label{sec:endpoints}
An endpoint's raw IP address is retained in the input snapshot but need not receive its own one-hot coordinate. For label predicates, the adapter resolves selected keys such as zone or role through the supplied inventory. Each selected key is a separate single-valued categorical field. Source and destination lookups include tenant and network identifiers, so reuse of an address in another network does not transfer labels. A snapshot declares a half-open validity interval in integer UTC seconds; observations outside it, missing endpoints and absent label keys produce missing categorical evidence. The adapter rejects malformed addresses, duplicate inventory endpoints and multi-valued labels. It does not infer an application role from a subnet or port.

Explicit address predicates follow a separate exact path. Let $\mathcal C$ contain every CIDR selector used for one endpoint field. Define
\[
a\sim b\ \Longleftrightarrow\ \forall c\in\mathcal C:\ [a\in c]=[b\in c].
\]
Addresses with identical membership signatures share one category. This is a one-field application of the predicate-equivalence principle used by Atomic Predicates~\cite{atomic}. The compiler sweeps the start and exclusive-end boundaries of all selectors, separately for IPv4 and IPv6, then merges intervals with identical signatures. It handles nested prefixes, unions of prefixes, and individual hosts represented by /32 or /128 selectors. A known address outside every selector receives a valid nonmatching class; an absent address remains missing.

\begin{auditproposition}[Exact address-class reduction]
Replacing each known address by its class and each rule's CIDR union by the classes it accepts preserves that field's predicate value.
\end{auditproposition}
\begin{proof}
Every selector's membership indicator is constant within a class by definition. The disjunction of any supplied selectors is therefore constant there as well. Empty predicates and wildcards retain their original definitions; missing addresses retain the missing sentinel.
\end{proof}
With $m$ distinct selectors, there are at most $2m+2$ boundary intervals across the two address families before identical signatures are merged. Lookup uses binary search over those intervals. The width depends on policy distinctions rather than the number of observed addresses, although many host-specific rules can still produce many classes. The current backend rejects more than 4096 classes in a field or total width above 16384. It uses no collision-prone feature hashing or arbitrary subnet coarsening.

Label predicates and address predicates remain distinct policy choices. A CIDR predicate tests the literal address; tenant or other contextual restrictions must be explicit in the policy. The inventory's network key scopes label resolution. The adapter emits the existing eight-field snapshot and records raw-policy, inventory and observation fingerprints alongside categorical fingerprints. Raw-address parsing and inventory indexing remain part of the pipeline and are measured separately from matching.

\subsection{Compilation and tiled correlation algorithm}
A policy snapshot fixes dictionaries, intended outcomes, tenant exemptions and hierarchy mappings. It compiles the weight matrix $\Theta$ and Boolean rule-to-framework reachability $B=[HG>0]$. The interface supports eight fields, at most 1024 requirements, 512 controls, 64 frameworks and total categorical width 16384. The implementation rebuilds a compiled snapshot when its policy changes; the dependency and re-audit costs are specified below.

Algorithm~\ref{alg:correlation} processes disjoint tiles. Static matrices are uploaded once. Each invocation uploads categorical inputs and outcomes, encodes each tile, computes one score product, applies the gates and accumulates the five vectors for each finding class. Rule totals are downloaded for final hit-count products; Boolean hierarchy reductions yield distinct-record totals. The timer stops only when all ten host vectors are available.

\begin{algorithm}[t]
\caption{Single-product hierarchical correlation audit}\label{alg:correlation}
\begin{algorithmic}[1]
\Require Compiled $\Theta,H,B$, intent $a$, exemptions $E$
\State Initialize ten 64-bit output vectors to zero
\For{disjoint tile $I$ of records}
  \State $F^+\gets$ extended one-hot encoding of $X_I$
  \State $\Xi\gets F^+\Theta^{\mathsf T}$
  \State $g_{ij}\gets[o_i\ne a_j]\land\neg E_{j,t(i)}$
  \State $V\gets[\Xi=0]\land g$
  \State $U\gets[0<\Xi<b]\land g$
  \For{$Z\in\{V,U\}$}
    \State $q_Z\gets q_Z+\mathbf{1}^{\mathsf T}Z$
    \State $d_Z^c\gets d_Z^c+\mathbf{1}^{\mathsf T}[ZH>0]$
    \State $d_Z^f\gets d_Z^f+\mathbf{1}^{\mathsf T}[ZB>0]$
  \EndFor
\EndFor
\For{$Z\in\{V,U\}$}
  \State $h_Z^c\gets q_ZH$; $h_Z^f\gets q_ZB$
\EndFor
\State \Return $(q_Z,h_Z^c,h_Z^f,d_Z^c,d_Z^f)$ for both classes
\end{algorithmic}
\end{algorithm}

\begin{auditproposition}[Complete report correctness]\label{prop:complete}
Under the snapshot and arithmetic assumptions, Algorithm~\ref{alg:correlation} returns the ten vectors defined in Section~\ref{sec:model}.
\end{auditproposition}
\begin{proof}
Proposition~\ref{prop:single} and the explicit gates produce exactly $V$ and $U$. Boolean reachability $B$ counts each rule--framework association once, regardless of path multiplicity. For any record, $[ZH>0]$ and $[ZB>0]$ indicate membership in the corresponding node's union of findings. Summing these indicators therefore counts distinct records, while summing $Z$ and multiplying its totals counts requirement hits. Proposition~\ref{prop:batches} permits summation over disjoint tiles. The finite-precision bounds preserve every intermediate integer.
\end{proof}

With tile size $T$, score work is $O(NR(D+K))$ and dense hierarchy work is $O(NR(C+Q))$. Static matrix storage is $O(R(D+K+C+Q))$; working storage is $O(T(D+K+R+C+Q))$, in addition to $O(NK)$ input storage. Tiling bounds incidence storage without changing the report. These are array-payload bounds, not measurements of runtime allocator overhead.

\subsection{Policy versions, changes and re-auditing}\label{sec:lifecycle}
A framework version alone does not determine the update rate: organizations can change predicates, exemptions and inventory evidence independently. We assume a snapshot is fixed during one audit, without assuming that custom policies change rarely. For $J$ equal-sized batches reusing a compiled snapshot, average cost is $T_{\mathrm{eval}}+T_{\mathrm{build}}/J$. The complete-JSON interface rebuilds on every request; it does not implement a persistent policy cache.

The dependencies distinguish three changes. A reporting-edge edit changes $H$, $G$ or $B$ but need not change predicate weights. A predicate, outcome or exemption edit changes applicability or its gate. A dictionary/CIDR-selector edit can additionally change record encodings; inventory updates can change resolved labels even with fixed dictionaries. The current implementation conservatively recompiles the full policy and backend for each supplied revision, including hierarchy-only edits; no incremental update speedup is claimed. Dense reachability preparation costs $O(RCQ)$ and the correlation weights occupy $O(R(D+K))$, besides validation and gate preparation.

An existing report remains a result for its recorded version. Restating historical logs under a new version incurs re-audit cost. For hierarchy-only changes, per-rule totals can regenerate hit totals, but cannot reconstruct distinct counts: incidence or the original records is needed. The implementation releases tile incidence, so it reruns the records for a complete revised report. Changed address classes or labels also require resolving the relevant records again.

A separate Apple M4 Pro study performs 420 measured full rebuilds across five policy shapes, three seeds and four edit types (Appendix~\ref{app:updates}). For a framework-edge edit with $D=128$, policy compilation plus host correlation-weight construction takes 4.38~ms at $(R,C,Q)=(256,32,4)$, 18.16~ms at $(1024,128,4)$, and 36.70~ms at $(1024,512,64)$. These are CPU preparation costs, excluding GPU upload, record resolution and historical replay; they must not be interpreted as complete GPU update latency.

\subsection{Direct predicates and two-product correlation}
Native C++ partitions records among four workers, applies the outcome/exemption gate, and evaluates each requirement's fields with early rejection on a definite false predicate. Unknown fields set a flag; a later false field still rejects the pair. Each worker accumulates rule totals and per-record Boolean flags for reachable controls and frameworks, then reduces its 64-bit totals. It does not materialize dense record--rule incidence.

The direct GPU path uses one thread per record and requirement pair. The two-product path implements $S=F\mathbf R^{\mathsf T}$ and $P=S+WL^{\mathsf T}$. Both emit the same three-state byte representation as the single-product path. All four GPU matchers, including the rule-bit-vector path below, share the same class-specific hierarchy products and reductions. This isolates the matching representation within GPU comparisons.
Two additional variants fuse each matrix path's score thresholds and outcome/exemption gates into one CUDA kernel. They preserve the same products and report reductions; Section~\ref{sec:fusion} compares both fused paths with the original implementations and the direct/bitmap controls.

\subsection{Compiled rule-bit-vector baseline}\label{sec:rulebits}
For every field value, including missing, precompute a bitmap over requirements that evaluate true and a bitmap over requirements that are not false. A wildcard sets both bits. A constrained empty set sets neither. Missing evidence under a constrained nonempty set sets only the possible bit. Outcome and tenant-exemption combinations similarly compile to gate bitmaps.

For each record, AND its $K$ true bitmaps and the gate bitmap to obtain $T_i$, and separately AND its possible bitmaps and the gate to obtain $P_i$. Confirmed requirements are $T_i$; indeterminate requirements are $P_i\mathbin{\&}\neg T_i$. This follows directly from three-valued conjunction: every field must be true for confirmation and none may be false for possibility. Unused bits in a final word remain zero.

The four-worker C++ version uses 64-bit words and enumerates set bits for the same per-record hierarchy unions as the direct native implementation. Its predicate stage costs $O(NK\lceil R/64\rceil)$ word operations, followed by work for surviving findings and their hierarchy edges. The GPU version uses 32-bit words: one warp handles one record and one requirement word; its first lane intersects the field masks, broadcasts the result, and each lane emits its own requirement's state. It then uses the common GPU hierarchy path. Bitmap storage is $O((D+K)\lceil R/w\rceil)$ words plus exemption gates and hierarchy mappings. These implementations provide an independently built flat bitmap baseline; they do not reproduce a published classifier's complete data structure.

\subsection{Snapshot interface and reference implementations}
The CLI validates a complete JSON snapshot and returns ten vectors, record and missing-record counts, content fingerprints, backend identity and optional bounded witnesses. A witness names the input record and requirement, observed and required outcomes, and unresolved field indices. Diagnostic witnesses use an independent CPU pass and are disabled in the timing study.

A NumPy reference packs confirmed and indeterminate membership over records, then performs bitwise unions and population counts for hierarchy nodes. This differs from the compiled rule bitmaps, whose bit positions are requirements. A scalar list/set interpreter provides an independent correctness oracle. Retained SQL and earlier packed-predicate measurements provide supplementary context~\cite{duckdb}; the main six-backend study concentrates on the new construction and competitive compiled alternatives.

\section{Evaluation Method}
The study tests correctness, one versus two score products, comparison with direct and bitmap matching, execution regimes, and the costs of complete inputs and changed policy snapshots.

\subsection{Correctness and provenance}
The scalar oracle evaluates original allowed-value lists and constructs hierarchy reachability using sets. Forty fixtures cover wildcards, constrained empty sets, exemptions, overlapping hierarchy paths, all-missing and empty inputs, exhaustive two-field predicates, 24 seeded cases, requirement-word boundaries and the 1024-rule limit. An additional score check compares every one of the $3^8=6561$ field-state vectors against the algebraic formula and all three thresholds. Together the fixtures check 157151 record--requirement scores. GPU tests also compare raw binary16 product scores with exact integer products. Six authored snapshots check complete reports, identifiers and witnesses across the six backends; duplicate identifiers are rejected. Twelve additional raw-endpoint fixtures check the new algorithm after CIDR and label resolution against an independent address/inventory interpreter.

Every warmup and measured invocation must match all ten reference vectors exactly. Raw results retain samples, seeds, input/policy fingerprints, output vectors, backend order and source hashes. The GPU driver checks a successful hardware-validation record and matching source hashes before timing. Profiling is separate from timing.

\subsection{Workloads and measurement scope}
Five main operating points span 10000 to one million records, 64 to 1024 requirements and categorical widths 128 and 1024. Each uses seeds 42, 123 and 456. Approximately half the fields are constrained, allowed-value density ranges from 0.2 to 0.8, two percent of field values are missing and three percent of tenant-exemption entries are set. Each requirement maps to two controls and each control to two frameworks. Both the categorical records and requirements in these synthetic workloads are generated. The distributions control shape and selectivity; they are not calibrated enterprise policies. Timings are measured executions on these generated inputs.

Each case uses two warmups and seven measured invocations, with backend order shuffled deterministically. Main tables report the median of the three seed medians. Prepared evaluation excludes static policy compilation/upload but includes dynamic input upload, encoding, matching, hierarchy aggregation and all ten host results. Setup costs and all repetitions are retained separately. The default GPU tile is 8192 records for every matcher.

The complete-snapshot experiment uses exactly the prepared 100000-record, 256-requirement, seed-123 input. Timing starts with serialized categorical JSON in host memory and includes decoding, validation, hashing, policy/backend construction, evaluation and output serialization. It excludes file/network I/O and diagnostic witnesses. A public five-rule replay tests a smaller policy on recorded observations. A separate CPU-only rebuild experiment changes framework edges, rule-to-control edges and predicates; Section~\ref{sec:lifecycle} and Appendix~\ref{app:updates} define its timing boundary.

Controlled ablations pad the same predicates to widths 512 and 2048 without changing outputs, change missingness to zero or twenty percent, and use GPU tiles of 512 and 32768. The baseline is width 128, two percent missingness and tile 8192. CPU runs of the tile cases are repeat controls because those evaluators do not use the tile argument. A separate retained endpoint experiment measures inventory compilation and raw-address resolution; its CPU timings are not compared as speedups against the GPU host.

\section{Results}\label{sec:results}

\subsection{Platforms and matched execution}

The six-backend study uses an AWS g4dn.xlarge with one NVIDIA T4 (15~GiB device memory), four Intel Xeon Platinum 8259CL vCPUs (two exposed cores with SMT) and 16~GiB host memory. Ubuntu 22.04 runs Python 3.10, NumPy 1.26.4, CuPy 13.3.0, GCC 11.4 and NVIDIA driver 595.91.07. The retained environment identifies the CUDA runtime and device library configuration. Native C++17 and the rule-bit-vector baseline both use \texttt{-O3} and four workers. The study contains 966 measured invocations, with all ten outputs checked each time. The separate Apple M4 Pro CPU study contributes 322 measurements using the same input fingerprints; CPU--GPU ratios below use only the T4 host.

\begin{table*}[t]\centering\small\caption{Prepared complete-audit latency on the shared T4 host, in milliseconds. Entries are medians of three seed medians, each based on seven measured invocations. Every GPU method includes encoding where applicable, dynamic transfers, matching, common hierarchy processing and host results. CPU methods use four workers.}\label{tab:algorithm-main}\begin{tabular}{rrr rrrrrr}\toprule $N$ & $R$ & $D$ & CPU direct & CPU bits & GPU direct & Two products & One product & GPU bits\\\midrule

10,000 & 64 & 128 & 3.62 & 0.76 & 2.46 & 3.75 & 3.66 & 2.46 \\

100,000 & 256 & 128 & 164.98 & 15.55 & 13.76 & 22.25 & 21.67 & 13.80 \\

100,000 & 1,024 & 128 & 775.00 & 60.38 & 30.20 & 62.26 & 60.56 & 26.88 \\

100,000 & 256 & 1,024 & 176.10 & 16.65 & 13.53 & 22.31 & 21.70 & 13.52 \\

1,000,000 & 256 & 128 & 1669.30 & 145.64 & 124.05 & 204.18 & 198.52 & 123.72 \\

\bottomrule\end{tabular}\end{table*}

At one million records and 256 requirements, single-product correlation takes 198.52~ms, compared with 204.18~ms for two products, 124.05~ms for direct GPU matching and 123.72~ms for GPU bitmaps. The corresponding CPU bitmap and direct times are 145.64 and 1669.30~ms. The CPU bitmap baseline improves on the direct loop by 11.46$\times$; comparing only with that loop would overstate the advantage relative to the stronger baseline.

The two-product/single-product latency ratios across the five main points range from 1.02 to 1.03$\times$. These are descriptive ratios of medians. Single-product correlation is slower than both direct GPU and GPU bitmap evaluation at every main point. The one-product construction preserves the audit semantics and changes execution boundaries; it does not establish dense correlation as the preferred matcher for these workloads.

On the M4 Pro, the one-million-record CPU direct/bitmap ratio is 10.40$\times$. This separately supports the importance of the compiled baseline, without combining timings from different hosts into an accelerator speedup.

\paragraph{Variation in the product-consolidation result}
Across the 15 seeded main cases, the single-product median is lower than the two-product median. Observed repetition ranges overlap in one case. Table~\ref{tab:variation} reports the ratio of medians of seed medians and conditional bootstrap intervals. We resample the three seed cases and, within each selected case, independently resample seven repetitions for each backend; 20000 replicates yield percentile intervals. The procedure describes the observed seeds and repetitions on this host. Three seeds and one execution session do not characterize variation across deployments or correlated timing drift.
\begin{table}[t]\centering\small\caption{Two-product/single-product latency ratios and conditional 95\% bootstrap percentile intervals from the retained T4 measurements. Values above one favor the single-product construction.}\label{tab:variation}\begin{tabular}{rrrcc}\toprule $N$ & $R$ & $D$ & Ratio & Interval\\\midrule
10,000 & 64 & 128 & 1.024 & [1.013, 1.036]\\
100,000 & 256 & 128 & 1.027 & [1.020, 1.034]\\
100,000 & 1,024 & 128 & 1.028 & [1.024, 1.031]\\
100,000 & 256 & 1,024 & 1.028 & [1.020, 1.037]\\
1,000,000 & 256 & 128 & 1.029 & [1.026, 1.033]\\
\bottomrule\end{tabular}\end{table}

\subsection{Pipeline and public-data costs}

\begin{table}[t]\centering\small\caption{Exact-input pairing at $N=100000$, $R=256$, $D=128$, seed 123 on the T4 host. Both columns are medians of seven invocations in milliseconds. Full includes JSON decoding, validation, hashing, compilation and serialization.}\label{tab:algorithm-pipeline}\begin{tabular}{lrr}\toprule Method & Prepared & Full snapshot\\\midrule

CPU direct & 163.42 & 836.29 \\

CPU bits & 15.19 & 704.07 \\

GPU direct & 13.76 & 677.59 \\

Two products & 22.25 & 680.67 \\

One product & 21.71 & 698.95 \\

GPU bits & 13.83 & 688.96 \\

\bottomrule\end{tabular}\end{table}

For that exact paired input, the CPU-bitmap/GPU-bitmap ratio changes from 1.10$\times$ prepared to 1.02$\times$ with full processing. The complete path rebuilds policy data structures each time; its extra cost therefore includes both common input work and backend-specific compilation. Small full-path differences are not interpreted as a stable ranking.

The public five-rule replay retains 65478 observations and produces confirmed rule counts $(17,144,0,0,7)$ and indeterminate counts $(0,0,0,38,0)$ across all methods. Prepared CPU direct and bitmap times are 3.09 and 2.25~ms. Single-product correlation takes 13.38~ms; GPU bitmaps take 8.52~ms. These measurements cover observed traffic under the explicitly hypothetical policy described in Section~\ref{sec:cases}.

\subsection{Representation and batching sensitivity}

Figure~\ref{fig:algorithm-ablation} changes vocabulary padding and tile size while holding logical inputs and outputs constant. Smaller tiles increase invocation and reduction overhead. Padding retains the exact predicates but adds dense coordinates, exposing work that direct predicates and rule-word intersections need not perform. The raw data also retain missingness ablations; those change applicability and early-rejection opportunities, rather than isolate layout alone.

\begin{figure*}[t]\centering\includegraphics[width=.9\textwidth]{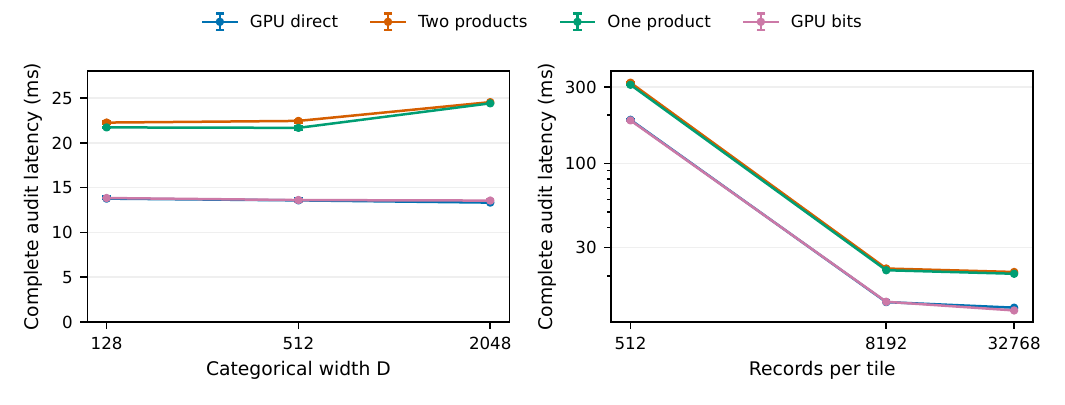}\caption{T4-host controlled ablations at 100000 records and 256 requirements. Medians of seven invocations; whiskers are observed minima and maxima, not confidence intervals. Left: unused-category padding, tile 8192. Right: tile size, width 128; logarithmic latency axis. Both horizontal axes are logarithmic. Logical records, policy predicates and all ten outputs remain fixed within each panel.}\label{fig:algorithm-ablation}\Description{Two plots compare direct GPU, two-product correlation, one-product correlation and GPU bitmap audit latency as unused vocabulary width or GPU tile size changes.}\end{figure*}

\subsection{Instruction evidence}

Nsight Compute profiles one separately validated audit with 8192 records, 256 requirements and width 128. The single-product path launches 44 kernels and the two-product path 46. Both paths execute 851,968 Tensor Core instructions at this shape, consistent with changing product boundaries rather than the total dense arithmetic. Their first matrix kernels record nonzero Tensor Core instruction counts; all four methods also use Tensor Core instructions in their shared hierarchy products. Direct matching is therefore not a Tensor-Core-free complete pipeline. Profiling establishes execution of those instructions, not utilization efficiency or a latency advantage. Instrumented runs are excluded from the timing tables.

\section{Fusion, Execution Regimes and an L4 Study}\label{sec:fusion}
The score construction exposes a second implementation choice: combine threshold tests and the outcome/exemption gate in one CUDA kernel. We apply this fusion to both matrix formulations, retain both original paths, and keep the direct and bitmap controls. All paths preserve the same ten vectors and GPU hierarchy operations. Fusion is an established implementation technique; this comparison isolates its interaction with the one-product encoding.

The eight-method study uses one AWS g6.xlarge with an NVIDIA L4 (23034~MiB visible device memory), four AMD EPYC 7R13 vCPUs (two cores with SMT), 16~GiB host memory, CUDA toolkit 12.9, CuPy 13.3.0 and driver 595.91.07. It repeats the five operating points above and tests 21 controlled regimes at 65536 records, plus the published ACL workload in Section~\ref{sec:stanford}. Each case has seeds 42, 123 and 456, two warmups and seven measured invocations per method. Method order is independently shuffled within each repetition. The 4536 measured invocations all match the reference vectors; separate stage diagnostics contribute 80 instrumented invocations. All latency comparisons within this study use the L4 host.

\begin{table*}[t]\centering\small
\caption{Prepared audit latency on the L4 host, in milliseconds: medians of three seed medians, each with seven interleaved repetitions. Both matrix variants fuse thresholds and outcome/exemption gates. Original one retains the unfused single-product path. CPU methods use four workers; complete raw results include CPU direct and unfused two-product execution.}\label{tab:fusion}
\begin{tabular}{lrrr rrrrrr}\toprule Workload & $N$ & $R$ & $D$ & CPU bits & GPU direct & GPU bits & Original one & Fused two & Fused one\\\midrule
Synthetic & 10,000 & 64 & 128 & 0.60 & 1.94 & 1.93 & 2.90 & 2.28 & 2.08\\
Synthetic & 100,000 & 256 & 128 & 12.34 & 10.68 & 10.76 & 16.77 & 12.94 & 11.68\\
Synthetic & 100,000 & 1,024 & 128 & 49.68 & 13.07 & 12.29 & 27.95 & 13.82 & 12.28\\
Synthetic & 100,000 & 256 & 1,024 & 13.16 & 10.74 & 10.70 & 17.00 & 13.19 & 11.82\\
Synthetic & 1,000,000 & 256 & 128 & 119.06 & 98.24 & 98.47 & 157.00 & 120.54 & 107.75\\
Published ACLs & 100,000 & 341 & 207 & 5.72 & 10.80 & 10.85 & 17.05 & 13.23 & 11.92\\
\bottomrule\end{tabular}\end{table*}

At 100000 records and 1024 requirements, fusion reduces single-product latency from 27.95 to 12.28~ms (2.28$\times$). Fused two-product execution takes 13.82~ms, so the one-product formulation retains an advantage after both paths receive the same optimization. GPU bitmaps take 12.29~ms. Their bitmap/single-product ratio has a conditional 95\% bootstrap interval of [0.988,1.009], consistent with comparable latency. The interval uses 20000 resamples of seeds and paired repetition indices; it describes this session rather than a population-wide performance guarantee.

The controlled regimes vary width (16--2048), requirements (64--1024), allowed-value density, four versus eight constrained fields, missingness and hierarchy fanout. Outcome/exemption gates are enabled throughout these controlled cases. Categories are sampled uniformly and allowed-set sizes round the requested density to an integer, with at least one allowed value. These are controlled stress cases, not calibrated enterprise distributions. Appendix~\ref{app:regimes} retains every regime, including a joint compact-domain/dense-predicate case.

GPU bitmaps have lower point medians in 20 of the 21 regimes; the remaining 1024-rule case is again comparable within the conditional interval. No tested regime establishes a robust overall correlation advantage over bitmaps. Separate CUDA-event diagnostics attribute 79.2--81.3\% of the sum of median fused-single stage intervals to hierarchy operations and reductions across the four profiled cases. These intervals include launch/idle gaps and are not pure kernel-work times or substitutes for uninstrumented latency. They explain why reducing score work alone has limited effect on the complete report.

The practical result is an exact matrix lowering whose execution improves substantially with fusion and can approach bitmap latency at larger rule counts. Compiled bitmaps outperform correlation in most tested regimes; the published ACL replay also favors CPU bitmaps. The measurements support choosing execution by workload and reporting cost, rather than inferring a speed advantage from the correlation formulation or Tensor Core use.

\paragraph{Hierarchy scaling with fixed predicates}
A separate T4 study holds $N=65536$, $R=256$, $D=128$, categorical records and predicates fixed within each seed and varies ten reporting graphs (Appendix~\ref{app:hierarchy}). All 1680 measured invocations reproduce their complete reports, with identical per-rule hits across graphs. Increasing controls from 32 to 512 at four frameworks changes fused-single latency from 10.87 to 10.99~ms; at 32 controls, increasing frameworks from four to 64 gives 9.22~ms. These are reporting-graph sensitivities with fixed applicability, not a claim of linear wall-clock scaling; kernel choice and fixed overhead can make dimension changes nonmonotonic. Comparisons remain within this host.

\subsection{Endpoint costs beyond the matcher}\label{sec:endpoint-results}
The retained Apple M4 Pro endpoint study uses 100000 observations per case, with two warmups and seven measurements per stage. With 16 fixed CIDR selectors, increasing distinct source IPs and inventory size from 100 to 100000 leaves 17 address classes and total width $D=36$. Increasing the selectors to 512 instead yields 513 classes and $D=532$. Thus exact compression removes dependence on observed address diversity in the first experiment, while retaining policy distinctions in the second.

At 100000 addresses and 16 selectors, full JSON-to-JSON processing takes 1595.60~ms and prepared four-worker direct evaluation 1.15~ms. Inventory preparation and raw-address resolution dominate this authored workload. Labels still require an inventory proportional to supplied endpoints. The performance cases use disjoint IPv4 /24 selectors; IPv6, nested prefixes, host selectors and stale/missing labels are separate correctness fixtures. No live inventory service is measured.

\section{Public Traffic and Policy Validation}\label{sec:cases}
We retain the UCI Internet Firewall Data source~\cite{uci}, including its CC BY 4.0 attribution and fingerprint. The source contains 65532 rows. Permit maps from \texttt{allow}; \texttt{deny} and \texttt{drop} map to not-permitted. We exclude \texttt{reset-both}, whose meaning is not equivalent to a binary initial decision, and explicitly report the excluded count.

Five hypothetical organization rules refer to destination ports 23, 3389, 445, 22 and 443. The port-22 requirement also depends on source zone, which the public records do not supply; its relevant observations remain indeterminate. Unavailable fields are left missing, not fabricated from port numbers. Port numbers do not establish a transport protocol or the actual application. The final requirement is an explicitly authored availability condition, not a conclusion derived from NIST text.

The replay is a reproducible test on public observations under stated hypothetical policies. The dataset does not include the source institution's approved policy or endpoint IP/label inventory, so it cannot validate IP resolution or establish that the institution was noncompliant. Source-row witnesses enable inspection of the transformation and result. This five-requirement replay exercises a small policy; the synthetic workloads cover larger rule sets.
\subsection{Public router policies with preserved priority}\label{sec:stanford}
The ATPG repository publishes a Stanford backbone configuration snapshot captured in June 2011~\cite{atpgdata}. From its 16 router configurations, we extract all 133 numbered IPv4 standard ACLs: 208 permit/deny entries, excluding remarks. These source-address filters provide policy-derived predicate structure beyond the hypothetical requirements used with the public firewall trace. Extended ACLs, interface attachment, routing and NAT are outside this adapter's scope.

An ACL uses first-match semantics and an implicit deny. Evaluating its original entries as independent obligations would change those semantics. For the source-address set $A_j$ of entry $j$ within one ACL, the adapter instead compiles
\[
 A_j^\star=A_j\setminus\bigcup_{\ell<j}A_\ell,
 \qquad A_{\mathrm{default}}=D_{\mathrm{IPv4}}\setminus\bigcup_j A_j.
\]
The original action labels each effective set; the default action is deny. These sets partition the known IPv4 inputs: any address belongs to its earliest matching entry or the default. Replacing them by exact selector-membership classes preserves that partition. This is a compatibility translation for the evaluation, not a new priority-classification algorithm. It requires a known address and ACL context; it does not define priority semantics for missing evidence.

The resulting snapshot has 341 obligations, 68 IPv4 address classes and total categorical width 207. Fields two and three hold ACL context and source-address class; a constant tenant category and five unconstrained singleton fields fill the unused positions without asserting source evidence. ACLs serve as report groups and routers as their parent groups, without assigned regulatory meaning. A direct interpreter checks the original address/wildcard predicates in source order. Across all 133 contexts and 281 address-boundary probes, its 37373 decisions and winning entries agree with the compiled partition; all ten report vectors are then checked independently of the correlation implementation.

The performance replay generates 100000 observations per seed. Half sample predicate boundaries or interiors and half sample the IPv4 space uniformly; ten percent of reference outcomes are independently flipped to inject reportable mismatches. Every invocation must reproduce the reference report. These are actual published policy predicates with generated observations, not paired operational traffic or evidence that the source network violated its policies. The separate UCI replay retains the paper's observed-traffic example.

\section{Limitations}
The single-product lowering reduces product boundaries while retaining dense arithmetic proportional to vocabulary width; compiled bitmaps can exploit word-level parallelism and sparse surviving findings. The T4 study covers the original six methods and the L4 study additionally evaluates fused matrix paths. These are comparisons within each host, not a scaling study across devices. Seven repetitions and three seeds provide descriptive latency evidence, not a universal ranking. All GPU methods use a common hierarchy path; a specialized bitmap aggregation pipeline could change the comparison.

Synthetic policies are not calibrated to enterprise selectivity or overlap. The public traffic replay has five hypothetical requirements and lacks the institution's approved policy and endpoint inventory. The router corpus supplies historical source-address ACLs, but its observations and outcome mismatches are generated; 1614 nonstandard numbered-ACL statements are outside the adapter. The endpoint study uses authored observations and supplied labels. Paired operational policies and observations would be needed to establish deployment relevance. The study does not measure analyst productivity, streaming recovery, multi-GPU scaling or peak resident memory.

The core accepts eight single-valued fields and binary outcomes. The standard-ACL adapter compiles source-address priority for known inputs; arbitrary extended priority policies, numeric ranges, multi-valued labels, temporal obligations and waiver workflows remain unsupported. Label validity intervals constrain supplied evidence but do not collect historical inventories. Complete JSON timings exclude file/network I/O and witness extraction. Fingerprints identify snapshots without authenticating the evidence or certifying regulatory compliance.

\section{Conclusion}
Retrospective rule auditing requires more than selecting a packet action: it must preserve uncertain applicability, all independent obligations and distinct-record unions across supplied reporting groups. The ternary correlation insight encodes false predicates with a dominating weight so that one product separates all three applicability states. We prove its composition with outcome gates and Boolean hierarchy aggregation, integrate exact address classes and labels, and evaluate implementations that return the same complete report. Matched fusion improves the matrix path, but compiled bitmaps remain faster in most tested regimes. Explicit rebuilding and re-audit boundaries expose the cost of changing policies. The result is an exact, evaluated matrix lowering with a defined scope and measured limits.
\label{sec:main-end}

\bibliographystyle{ACM-Reference-Format}
\bibliography{paper_references}
\clearpage
\appendix
\raggedbottom
\section{Supported Attribute Schema and Relationships}\label{app:schema}
The schema describes the implemented adapter and the supplied examples. It is a catalog of supported inputs, not a universal ontology of network attributes. A snapshot selects eight single-valued fields; only the first position, tenant, is fixed. The remaining positions select observation values, endpoint labels or address classes. Table~\ref{tab:field-kinds} gives their interpretation. Named vocabularies are supplied with the policy; the categorical engine receives integer identifiers and a separate missing sentinel.

\begin{table}[H]\centering\small
\caption{Supported field kinds. All field predicates use wildcard, empty-set or allowed-set semantics from Section~\ref{sec:model}.}\label{tab:field-kinds}
\begin{tabular}{@{}p{.12\columnwidth}p{.40\columnwidth}p{.38\columnwidth}@{}}\toprule
Kind & Source and representation & Missingness and scope\\\midrule
Tenant & Declared observation tenant; explicit vocabulary & Missing has its own exemption entry.\\
Value & Named observation scalar, e.g., protocol or service; explicit vocabulary & Absent key or null is missing; undeclared values are rejected.\\
Label & Source/destination label, e.g., zone or role; explicit vocabulary & Missing or out-of-date evidence stays missing; lookup is scoped by tenant and network.\\
CIDR & Source/destination address; exact selector-membership class, including host selectors & Absent differs from known nonmatch; literal prefix membership does not implicitly restrict context.\\\bottomrule
\end{tabular}\end{table}

\subsection{Concrete field layouts}
The authored endpoint example in the artifact uses the ordered fields in Table~\ref{tab:example-fields}. Field names describe that policy, not fixed kernel positions. A supplied service name is not inferred from a port. Each selected label key occupies its own field; multi-valued labels are rejected.

\begin{table}[H]\centering\small
\caption{One-based positions in the authored endpoint example.}\label{tab:example-fields}
\begin{tabular}{rll}\toprule
Position & Field & Kind / selected source\\\midrule
1 & tenant & tenant\\
2 & protocol & value / protocol\\
3 & source IP & CIDR / source\\
4 & destination IP & CIDR / destination\\
5 & service & value / service\\
6 & source zone & label / source zone\\
7 & destination zone & label / destination zone\\
8 & destination role & label / destination role\\\bottomrule
\end{tabular}\end{table}

The public firewall replay uses a different declared layout: position five holds destination-port categories 23, 3389, 445, 22, 443 and other. Position three is the source-zone predicate used by one requirement; its evidence is unavailable. All seven fields other than destination port remain missing. This replay supplies neither endpoint addresses nor an inventory. The retained real-flow example applies the same port encoding to one observation; it does not fill the missing fields with inferred labels.

The standard-ACL adapter uses position two for ACL context and position three for the source-address class. The other positions are unconstrained placeholders, including a constant tenant category. These placeholders do not infer observation attributes. The adapter requires known IPv4 addresses and context; its generated observations are separate from the public traffic replay.

\subsection{Context and reporting relations}
Record identifiers are unique within the input and support distinct-record counts and witnesses; they are not predicate coordinates. Observed outcomes and required outcomes are binary and form the mismatch gate. Tenant exemptions are a separate requirement--tenant relation. Observation timestamps select whether the supplied inventory is valid within its half-open interval; they are not temporal rule predicates. Inventory versions and SHA-256 fingerprints identify supplied evidence. They do not assert label authenticity.

Address containment, inventory lookup and reporting edges have different meanings. Exact CIDR classes preserve all supplied selector memberships, including nested prefixes and host exceptions. Labels are supplied facts scoped by tenant and network and are not automatically derived from subnet containment. Requirements map many-to-many to controls and frameworks; Boolean reachability removes duplicate paths before distinct-record unions. These relationships do not imply a shared attribute hierarchy.

\subsection{Domain boundaries}
Finite domains apply to a policy snapshot, not to an enumeration of every possible flow. Encoding width is the sum of field-domain sizes, with one extra missing coordinate per field in the ternary encoding. Address-class width depends on the supplied selectors; inventories still grow with their endpoint populations. Updating a dictionary or selector set requires recompilation. The implementation bounds and exact address-class construction are specified in Section~\ref{sec:endpoints}. Arbitrary numeric ranges, multi-valued labels and temporal obligations require further compilation or semantics and are outside this implementation.

\clearpage
\section{Complete Controlled Regime Results}\label{app:regimes}
All 21 regimes use 65536 records and three seeds. Except where varied, $R=256$, $D=128$, four constrained fields, allowed-value density $a=0.5$, missingness $m=0.02$ and two control edges per rule. Here $k$ is the number of constrained fields. The compact-dense case uses $D=16$, $k=8$, $a=1$ and $m=0$. Every method returns the complete ten-vector audit.
\begin{table}[H]\centering\small\caption{All controlled L4 regimes. CPU and GPU columns are bitmaps; Two and One are fused matrix variants. Milliseconds, medians of three seed medians.}\label{tab:all-regimes}
\begin{tabular}{@{}lrrrr@{}}\toprule Regime & CPU & GPU & Two & One\\\midrule
base & 10.03 & 6.69 & 8.11 & 7.30\\
$D=16$ & 9.78 & 6.80 & 8.19 & 7.35\\
$D=32$ & 9.97 & 6.78 & 8.16 & 7.36\\
$D=64$ & 9.99 & 6.79 & 8.14 & 7.35\\
$D=512$ & 10.14 & 6.75 & 8.20 & 7.38\\
$D=2048$ & 10.18 & 6.70 & 8.28 & 7.42\\
$R=64$ & 3.19 & 6.67 & 8.06 & 7.21\\
$R=1024$ & 41.78 & 8.33 & 9.35 & 8.27\\
$a=0.25,k=4$ & 3.46 & 6.73 & 8.12 & 7.33\\
$a=0.9,k=4$ & 68.66 & 6.75 & 8.20 & 7.31\\
$a=1.0,k=4$ & 113.88 & 6.76 & 8.16 & 7.34\\
$a=0.25,k=8$ & 2.57 & 6.77 & 8.15 & 7.34\\
$a=0.5,k=8$ & 3.54 & 6.79 & 8.19 & 7.35\\
$a=0.9,k=8$ & 41.60 & 6.79 & 8.17 & 7.32\\
$a=1.0,k=8$ & 113.98 & 6.76 & 8.22 & 7.35\\
$m=0.0$ & 9.32 & 6.71 & 8.13 & 7.37\\
$m=0.2$ & 17.63 & 6.77 & 8.17 & 7.33\\
$m=0.8$ & 77.06 & 6.75 & 8.14 & 7.44\\
fanout 1 & 7.80 & 6.72 & 8.16 & 7.29\\
fanout 8 & 16.39 & 6.76 & 8.15 & 7.28\\
compact dense & 113.81 & 6.71 & 8.12 & 7.56\\
\bottomrule\end{tabular}\end{table}
\begin{table}[H]\centering\small\caption{GPU-bitmap/fused-single ratios. Ratios above one favor correlation. Conditional paired-bootstrap intervals are descriptive, without a population or multiple-comparison guarantee.}\label{tab:regime-intervals}
\begin{tabular}{@{}lrr@{}}\toprule Regime & Ratio & Conditional 95\% interval\\\midrule
base & 0.917 & [0.912,0.928]\\
$D=16$ & 0.925 & [0.907,0.929]\\
$D=32$ & 0.921 & [0.901,0.941]\\
$D=64$ & 0.924 & [0.909,0.938]\\
$D=512$ & 0.915 & [0.901,0.922]\\
$D=2048$ & 0.903 & [0.893,0.937]\\
$R=64$ & 0.925 & [0.900,0.930]\\
$R=1024$ & 1.007 & [0.996,1.015]\\
$a=0.25,k=4$ & 0.919 & [0.908,0.933]\\
$a=0.9,k=4$ & 0.924 & [0.908,0.941]\\
$a=1.0,k=4$ & 0.921 & [0.896,0.938]\\
$a=0.25,k=8$ & 0.923 & [0.901,0.942]\\
$a=0.5,k=8$ & 0.924 & [0.911,0.941]\\
$a=0.9,k=8$ & 0.927 & [0.913,0.931]\\
$a=1.0,k=8$ & 0.921 & [0.903,0.942]\\
$m=0.0$ & 0.910 & [0.898,0.937]\\
$m=0.2$ & 0.923 & [0.905,0.948]\\
$m=0.8$ & 0.907 & [0.890,0.948]\\
fanout 1 & 0.921 & [0.898,0.926]\\
fanout 8 & 0.928 & [0.917,0.945]\\
compact dense & 0.887 & [0.879,0.938]\\
\bottomrule\end{tabular}\end{table}
The artifact retains direct CPU/GPU and both unfused matrix timings, every sample and method order, setup measurements and all output vectors. Setup timings describe these implementations after correctness preflight; they do not establish a general policy-compilation advantage.

\clearpage
\section{Policy Rebuild Experiment}\label{app:updates}

This CPU-only study measures the existing full-rebuild path on an Apple M4 Pro with 14 physical cores and 48~GiB memory, Python 3.12 and NumPy 2.5.3. BLAS thread settings are four. It is separate from the T4/L4 execution studies; preparation costs are not combined with timings from those hosts into an amortized latency claim.

Five $(R,C,Q,D)$ shapes use three seeds and four changes: rebuilding the unchanged document (U), editing one control-to-framework edge (F), editing one rule-to-control edge (H), and toggling one field value in $\lceil0.01R\rceil$ requirements (P). Edited documents are constructed before timing. Cases are shuffled deterministically; each has two warmups and seven measured rebuilds. Each rebuild times full \texttt{Policy.compile}, then host correlation-weight construction and CPU bitmap-runner construction in separate sequential stages. These are current implementation costs, not optimized incremental algorithms or a claim that one representation has an intrinsic compilation advantage.

The 420 measured rebuilds provide 1260 stage intervals. After every warmup and measured rebuild, both constructed implementations evaluate 33 generated records and match all ten vectors from the independent scalar oracle (1080 complete-report comparisons). Compilation includes policy validation, canonicalization, fingerprints and dense Boolean reachability. GPU initialization/upload, record encoding, historical replay and editing the document are excluded. No deployment update-frequency distribution is assumed.

\begin{table}[H]\centering\small\setlength{\tabcolsep}{3pt}\caption{CPU preparation after policy changes, milliseconds. Medians of three seed medians, each with seven samples. Compile includes common policy preparation; Weights and Bits are additional host construction stages. Total is measured compile-plus-weights, not the sum of separately aggregated medians.}\label{tab:policy-updates}

\begin{tabular}{@{}lcrrrr@{}}\toprule $R/C/Q/D$ & Edit & Compile & Weights & Bits & Total\\\midrule

64/8/4/128 & F & 0.62 & 0.51 & 1.62 & 1.12\\

64/8/4/128 & P & 0.58 & 0.50 & 1.58 & 1.10\\

64/8/4/128 & H & 0.57 & 0.50 & 1.64 & 1.10\\

64/8/4/128 & U & 0.57 & 0.50 & 1.60 & 1.07\\

256/32/4/128 & F & 2.40 & 1.99 & 6.27 & 4.38\\

256/32/4/128 & P & 2.84 & 2.29 & 6.98 & 5.13\\

256/32/4/128 & H & 2.43 & 2.05 & 6.39 & 4.51\\

256/32/4/128 & U & 2.37 & 1.98 & 6.34 & 4.33\\

256/32/4/1024 & F & 15.09 & 2.09 & 6.93 & 17.29\\

256/32/4/1024 & P & 15.19 & 2.08 & 6.90 & 17.36\\

256/32/4/1024 & H & 15.23 & 2.07 & 6.94 & 17.38\\

256/32/4/1024 & U & 15.12 & 2.17 & 7.02 & 17.24\\

1024/128/4/128 & F & 9.96 & 8.23 & 25.41 & 18.16\\

1024/128/4/128 & P & 9.87 & 8.08 & 25.17 & 17.96\\

1024/128/4/128 & H & 9.78 & 8.10 & 25.36 & 17.89\\

1024/128/4/128 & U & 9.85 & 8.19 & 25.17 & 18.02\\

1024/512/64/128 & F & 28.73 & 8.15 & 25.98 & 36.70\\

1024/512/64/128 & P & 28.79 & 8.28 & 26.19 & 37.10\\

1024/512/64/128 & H & 28.96 & 8.16 & 26.16 & 37.09\\

1024/512/64/128 & U & 28.89 & 8.20 & 26.38 & 37.16\\

\bottomrule\end{tabular}\end{table}

\newpage
\section{Hierarchy Scaling with Fixed Predicates}\label{app:hierarchy}

This separate eight-method study uses an AWS g4dn.xlarge with an NVIDIA T4, four vCPUs and 16~GiB nominal host memory. Software is Python 3.10, NumPy 1.26.4 and CuPy 13.3.0 with CUDA toolkit 12.9. Device, CPU, driver and package details are retained with the raw results. It does not combine this host with either earlier GPU session into a speedup ratio.

All cases use $N=65536$, $R=256$, $D=128$, four constrained fields, allowed density 0.5 and missingness 0.02. Within each seed, categorical records, predicates, outcomes and exemptions are unchanged; only the reporting edges change. The graph parameters below are $(C,Q,e,f)$, where every rule has $e$ control edges and every control has $f$ framework edges. The one-framework/eight-control-edge case creates eight paths from each rule to that framework; Boolean reachability preserves one association.

Ten graphs and three seeds produce 30 cases. Eight methods use two warmups and seven timed invocations each, with method order independently shuffled per repetition: 1680 measured complete-audit invocations. Cases run in seed order and the listed graph order; intervals between graph cases and hardware variability limit interpretation of small cross-case differences. Each checks all ten vectors, and per-rule hits remain identical across graphs. Before timing, 30 smaller cases are compared against the independent scalar/set interpreter, including maximum tested control/framework counts; a separate hand-checked fixture checks duplicate paths and shared records. All six GPU variants and three CPU/reference implementations pass those checks.

\begin{table}[H]\centering\small\setlength{\tabcolsep}{3pt}\caption{Prepared audit latency by hierarchy shape on one T4 host. Milliseconds, medians of three seed medians. CPU and GPU bits are compiled bitmaps; Two and One use matched fused matrix epilogues.}\label{tab:hierarchy-scaling}

\begin{tabular}{@{}lrrrrr@{}}\toprule $(C,Q,e,f)$ & CPU bits & GPU direct & GPU bits & Two & One\\\midrule

32/4/2/2 & 11.78 & 10.00 & 9.93 & 12.01 & 10.87\\

8/4/2/2 & 11.74 & 9.32 & 9.24 & 10.54 & 9.92\\

128/4/2/2 & 13.96 & 9.34 & 8.78 & 11.22 & 11.04\\

512/4/2/2 & 21.86 & 11.87 & 10.77 & 11.36 & 10.99\\

32/1/2/1 & 8.84 & 9.29 & 9.61 & 11.08 & 10.67\\

32/16/2/2 & 11.84 & 8.64 & 8.77 & 10.19 & 9.51\\

32/64/2/2 & 12.38 & 9.30 & 9.17 & 11.13 & 9.22\\

32/4/8/2 & 15.26 & 8.83 & 8.91 & 11.01 & 9.79\\

32/4/2/4 & 11.07 & 9.96 & 10.01 & 12.05 & 10.86\\

32/1/8/1 & 12.76 & 9.33 & 9.25 & 11.24 & 10.09\\

\bottomrule\end{tabular}\end{table}

Graph edits change the hierarchy outputs; they are not claimed to preserve the entire report between cases. The invariant is the input to matching and its per-rule confirmed/indeterminate hits. Varying $C$ or $Q$ changes dense product dimensions; varying fanout at fixed dimensions changes overlap and sparse work without reducing the dense arithmetic. These generated mappings are not a measured distribution of regulatory crosswalks. Full CPU-direct, unfused-matrix and setup timings remain in the raw data.

\clearpage
\end{document}